\documentclass[11pt,a4paper]{article}
\usepackage[utf8]{inputenc}
\usepackage{amsmath,amssymb,amsfonts}
\usepackage{graphicx}
\usepackage{booktabs}
\usepackage[margin=2.5cm]{geometry}
\usepackage{hyperref}
\usepackage{amsthm}

\newcommand{\tr}{\mathrm{Tr}}
\newcommand{\Ztwo}{Z_2}

\newtheorem{proposition}{Proposition}
\newtheorem{theorem}{Theorem}
\newtheorem{conjecture}{Conjecture}

\title{Parity superselection obstructs monogamy of mutual information in free fermions}
\author{A.~Sokolovs}
\date{March 2026}

\begin{document}
\maketitle

\begin{abstract}
We prove that free fermions in the spin (tensor product) factorization
violate monogamy of mutual information: $I_3^{\mathrm{spin}} > 0$
for three adjacent strips of width $w = 1, 2, 3, 4, 5$ at all Fermi momenta
(certified with margin-to-error ratios from $1.9\times 10^5$ down to $73$; Table~1),
and for all~$w$ at $z = k_F w < z^* \approx 1.329$.
Many-body computation at $w = 6$ (using the $G$-matrix formula for spin-basis
matrix elements) maps the full scaling-limit function $I_3^{\mathrm{spin}}(z)$:
it has a global minimum of $0.100$ at $z \approx 1.5$ and grows monotonically
to $0.135$ at half filling (Fig.~\ref{fig:scaling}), establishing $I_3^{\mathrm{spin}} > 0$
for all~$z$ in the large-$w$ limit.
The proof rests on an exact operator identity---the fermionic and
spin reduced density matrices of disjoint regions differ by the
parity insertion $(-1)^{N_B}$ in the partial trace---and a rigorous
entropy bound. A separate Perron--Frobenius argument proves that the
parity insertion always damps the off-diagonal matrix elements;
for free fermions, an independent Gaussian bound then gives the
entropy ordering $\Delta S_{AD} \ge 0$. Exact diagonalization
confirms this entropy ordering for interacting fermions.

DMRG calculations on the $t$-$V$ chain quantify the effect for
interacting fermions: the factorization contribution to the
apparent $K$-dependence of $I_3$ exceeds the genuine interaction
contribution by a factor of~8 at moderate filling, and accounts
for ${\sim}80\%$ of the deviation observed in spin-basis numerics.
Strong repulsion ($K \lesssim 0.7 \pm 0.1$) restores monogamy in both
algebras.

Conversely, $Z_2$ parity superselection on $\rho_{AD}$ enforces
$I_3 \le 0$ at all fillings (proved for $w \le 3$), with the
ratio of parity entropy to quantum excess approaching the exact
$w$-independent value $2\ln 2/(3\ln\frac{4}{3}) = 1.606$.

These results imply that any use of $I_3$ as a diagnostic---whether
for holographic duality, quantum chaos, or Fermi surface
topology---must specify the operator algebra; without this
specification, the sign of $I_3$ is ambiguous.
\end{abstract}

\section{Introduction}

The tripartite information
\begin{equation}
I_3(A{:}B{:}D) = S_A + S_B + S_D - S_{AB} - S_{BD} - S_{AD} + S_{ABD}
\end{equation}
quantifies irreducible tripartite correlations~\cite{Hayden2013}. States satisfying $I_3 \le 0$ for all subsystems obey monogamy of mutual information (MMI)---a property guaranteed by holographic duality~\cite{Hayden2013,Wall2014} but violated by gapless free field theories~\cite{Casini2009,Agon2022}.

For lattice fermions, entanglement depends on the choice of operator algebra~\cite{Wiseman2003,Zanardi2002,Banuls2007,Szalay2021}, although for bipartite entanglement entropy the effect of superselection is sub-leading~\cite{KlichLevitov2008}. Two natural choices exist: the spin (tensor product) factorization $\mathcal{H} = \bigotimes_i \mathbb{C}^2$, where partial traces are the standard ones; and the fermionic (CAR) factorization, where subsystem entropies follow from the Peschel formula~\cite{Peschel2003,Eisler2009} and the partial trace respects parity superselection~\cite{Fagotti2010,CasiniHuerta2009,Maric2024}. For contiguous regions the two coincide. For disjoint regions they do not.

In~\cite{paper1} we showed that the fermionic $I_3$ for three adjacent strips of width~$w$ is a universal function $g(z)$ of $z = k_F w$, with a unique zero at $z^*\approx 1.329$. In this paper we prove that switching to the spin factorization changes the sign of~$I_3$: the quantity $I_3^{\mathrm{spin}}$ is strictly positive for $w = 1,\ldots,5$ at all~$z$, and for all~$w$ at $z < z^*$ (Fig.~\ref{fig:main}).

\begin{figure}[!ht]
\centering
\includegraphics[width=0.65\textwidth]{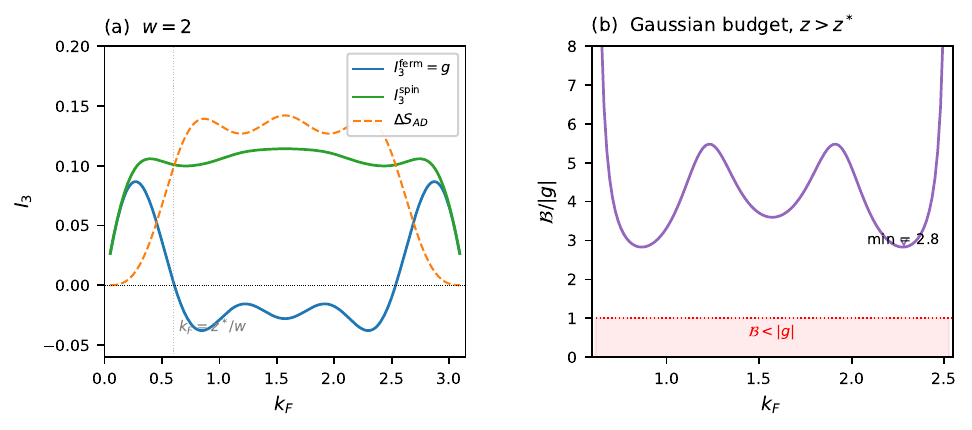}
\caption{\label{fig:main}
(a)~The fermionic $I_3^{\mathrm{ferm}} = g(z)$ (blue) changes sign at
$k_F \approx 0.60$ (the $w = 2$ zero; the scaling-limit value $z^*/w \approx 0.66$ is reached as $w \to \infty$), while $I_3^{\mathrm{spin}}$ (green) remains positive
for all~$k_F$. The difference $\Delta S_{AD}$ (orange dashed)
overwhelms $|g|$ wherever $g < 0$.
(b)~In the region where $g < 0$: the Gaussian budget $\mathcal{B}$
(Theorem~\ref{thm:budget}, purple) exceeds $|g|$ (red shaded).
Data for $w = 2$ (representative; $w = 1, 3$ show the same
qualitative behavior with margins in Table~1).}
\end{figure}

The mechanism is a single algebraic object. Proposition~1 shows that the fermionic and spin reduced density matrices $\rho_{AD}^{\mathrm{ferm}}$, $\rho_{AD}^{\mathrm{spin}}$ share the same $D$-parity-preserving matrix elements, but differ in the $D$-parity-changing sector: there, the fermionic matrix elements equal those of a parity-twisted partial trace $\rho_{AD}^{\mathrm{tw}} = \tr_B[(-1)^{N_B}\rho_{ABD}]$. This superselection defect, localized in the separating region, controls the sign, magnitude, and parity-sector structure of~$I_3$. Every result in this paper follows from this single identity.

\section{Setup}

Free fermions on a chain of $L$ sites at Fermi momentum~$k_F$, partitioned into three adjacent blocks $A$, $B$, $D$ of equal width~$w$. The scaling variable is $z = k_F w$.

Using the equal-width symmetry $S_A = S_B = S_D \equiv S_1$, $S_{AB} = S_{BD} \equiv S_2$, $S_{ABD} \equiv S_3$:
\begin{equation}
I_3 = \underbrace{(S_1 - 2S_2 + S_3)}_{\Delta^2 S} + \underbrace{(2S_1 - S_{AD})}_{I(A:D)}.
\end{equation}

This is a special case of the general identity $I_m = (-1)^{m+1}\Delta^m S$ for $m$-partite information~\cite{paper2}, decomposed into contiguous and non-contiguous parts. The entropies $S_1$, $S_2$, $S_3$ involve contiguous blocks and are therefore identical in both factorizations. The entire algebra dependence resides in~$S_{AD}$.

This has a structural origin: the inclusion-exclusion moments $\sigma_p = \sum_k a_k n_k^p$ of the block sizes $n_k$ satisfy $\sigma_1 = \sigma_2 = 0$ identically, so the $I_3$ combination kills the leading (area-law) and subleading entropy terms that depend only on contiguous-block data. The first surviving moment is $\sigma_3 = 6w^3$, which couples to the disjoint-block entropy $S_{AD}$---the sole term sensitive to whether $A\cup D$ is treated as contiguous or separated:
\begin{equation}\label{eq:delta}
I_3^{\mathrm{spin}} - I_3^{\mathrm{ferm}} = S_{AD}^{\mathrm{ferm}} - S_{AD}^{\mathrm{spin}} \equiv \Delta S_{AD}.
\end{equation}

\section{The superselection defect identity}

Both $\rho_{AD}^{\mathrm{spin}}$ and $\rho_{AD}^{\mathrm{ferm}}$ are obtained by tracing $\rho_{ABD}$ over~$B$, but in different factorizations.

\textbf{Spin trace.} $\rho_{n_A n_D,\, n'_A n'_D}^{\mathrm{spin}} = \sum_{n_B}\langle n_A n_B n_D|\rho_{ABD}|n'_A n_B n'_D\rangle$.

\textbf{Fermionic trace.} The JW string through $B$ contributes $(-1)^{N_B}$ for each $D$ creation operator anticommuted past~$B$. The bra contributes $(-1)^{N_B\cdot N_D}$ and the ket $(-1)^{N_B\cdot N'_D}$. Since $N_D - N'_D$ and $N_D + N'_D$ have the same parity (their difference $2N'_D$ is even), the net sign is:
\begin{equation}
(-1)^{N_B(N_D + N'_D)} = \begin{cases} +1 & \text{if } N_D \equiv N'_D \pmod{2}, \\ (-1)^{N_B} & \text{otherwise.}\end{cases}
\end{equation}

Define the parity-twisted partial trace $\rho_{AD}^{\mathrm{tw}} = \tr_B[(-1)^{N_B}\rho_{ABD}]$.

\begin{proposition}[Superselection defect identity]
For any state on adjacent blocks $A$, $B$, $D$:
\begin{equation}\label{eq:prop1}
\rho_{ij}^{\mathrm{ferm}} = \begin{cases} \rho_{ij}^{\mathrm{spin}} & \text{if } (-1)^{N_D^{(i)}} = (-1)^{N_D^{(j)}}, \\ \rho_{ij}^{\mathrm{tw}} & \text{if } (-1)^{N_D^{(i)}} \ne (-1)^{N_D^{(j)}}.\end{cases}
\end{equation}
Equivalently, as a single operator equation:
\begin{equation}\label{eq:block}
\rho_{AD}^{\mathrm{ferm}} = \tfrac{1}{2}(\rho_{AD}^{\mathrm{spin}} + \rho_{AD}^{\mathrm{tw}}) + \tfrac{1}{2}\Gamma_D(\rho_{AD}^{\mathrm{spin}} - \rho_{AD}^{\mathrm{tw}})\Gamma_D,
\end{equation}
where $\Gamma_D = (-1)^{N_D}$ is the $D$-parity grading operator.
\end{proposition}

Equation~\eqref{eq:block} is equivalently a block decomposition in the $D$-parity basis: the $D$-parity-preserving blocks of $\rho^{\mathrm{ferm}}$ equal those of $\rho^{\mathrm{spin}}$, while the $D$-parity-changing blocks equal those of $\rho^{\mathrm{tw}}$. The defect $\rho^{\mathrm{spin}} - \rho^{\mathrm{tw}}$ measures the difference between the two factorizations. When $\rho^{\mathrm{tw}} = \rho^{\mathrm{spin}}$ (no parity interference), the defect vanishes and $\rho^{\mathrm{ferm}} = \rho^{\mathrm{spin}}$---the two algebras agree. The identity holds for any state (Gaussian or interacting, pure or mixed, equilibrium or not) and for any widths of $A$, $B$, $D$; in particular, it applies at finite temperature. Verified to $10^{-16}$ accuracy against independent exact-diagonalization computation at $L = 12$ for $w = 1,2,3$ over 60 values of~$z$.

Three consequences are immediate.

\emph{Identical diagonals.} $\rho^{\mathrm{ferm}}$ and $\rho^{\mathrm{spin}}$ have the same diagonal elements (occupation probabilities), since $\delta_D = 0$ on the diagonal.

\emph{Coherence damping.} For $D$-parity-changing elements, the parity insertion produces destructive interference: $\rho_{ij}^{\mathrm{spin}} = \sum_b p_b X_b$ while $\rho_{ij}^{\mathrm{tw}} = \sum_b(-1)^{N_b}p_b X_b$ (the physical picture behind this decomposition is developed in Section~\ref{sec:physical}). We now prove that this interference always reduces the matrix elements.

\begin{theorem}[Coherence damping]\label{thm:damping}
Let $H$ be a fermionic Hamiltonian with real non-positive hopping ($H_{\alpha\beta} \le 0$ for $\alpha\ne\beta$ in the occupation basis). For the ground state and any tripartition into adjacent blocks $A$, $B$, $D$:
\begin{equation}
|\rho_{ij}^{\mathrm{tw}}| \le \rho_{ij}^{\mathrm{spin}} = |\rho_{ij}^{\mathrm{spin}}|
\end{equation}
for all $D$-parity-changing matrix elements of $\rho_{AD} = \mathrm{Tr}_{BE}|\psi\rangle\langle\psi|$.
\end{theorem}

\begin{proof}
The many-body Hamiltonian satisfies $H_{\alpha\beta}\le 0$ for $\alpha\ne\beta$ in the occupation basis (hopping terms contribute $-t$; interaction and on-site terms are diagonal). By the Perron--Frobenius theorem, the ground state can be chosen with all amplitudes non-negative: $\psi(\mathrm{config})\ge 0$. The $D$-parity-changing elements decompose by $B$-occupation number~$m$: $\rho^{\mathrm{spin}}_{ij} = \sum_m S_m$ and $\rho^{\mathrm{tw}}_{ij} = \sum_m(-1)^m S_m$, where $S_m = \sum_{b:|b|=m,\,e}\psi(a_i,b,d_i,e)\,\psi(a_j,b,d_j,e) \ge 0$ since each term is a product of non-negative numbers. Therefore $\rho^{\mathrm{spin}}_{ij} = \sum_m S_m \ge 0$ and $|\rho^{\mathrm{tw}}_{ij}| = |\sum_m(-1)^m S_m| \le \sum_m S_m = \rho^{\mathrm{spin}}_{ij}$.
\end{proof}

The theorem covers the $t$-$V$ chain at any coupling, any filling, and any strip width~$w$; it requires only that the hopping matrix elements be real and non-positive, a condition satisfied by bipartite lattice fermion models without magnetic flux or spin-orbit coupling. Systems with complex hopping (magnetic field), frustrated lattices (triangular, kagome), or non-bipartite geometries are \emph{not} covered: there the ground-state amplitudes can have mixed signs and the coherence damping inequality may fail. Verified to $10^{-14}$ accuracy by exact diagonalization at $L = 8$--12 for $V\in[-1.5,1.5]$.

\emph{Remark (pure states).} For a pure state with fixed particle number~$N$, the inequality becomes exact equality: $|\rho_{ij}^{\mathrm{tw}}| = |\rho_{ij}^{\mathrm{spin}}|$. The proof is one line: both $\psi(a_i, b)$ and $\psi(a_j, b)$ vanish unless $|b| = N - k$ where $k = |{\mathrm{occ}}_A| + |{\mathrm{occ}}_D|$, so $(-1)^{N_b}$ is a constant factor and $\rho^{\mathrm{tw}}_{ij} = (-1)^{N-k}\rho^{\mathrm{spin}}_{ij}$. The strict inequality arises only for mixed states ($\rho_{ABD} = \mathrm{Tr}_E|\psi\rangle\langle\psi|$), where particle-number fluctuations in~$ABD$ allow multiple~$|b|$ values to contribute with alternating signs.

\emph{Entropy ordering.} $\rho^{\mathrm{ferm}}$ and $\rho^{\mathrm{spin}}$ share the same diagonal elements and the same $D$-parity-preserving off-diagonal blocks (Proposition~1). Theorem~\ref{thm:damping} guarantees $|\rho^{\mathrm{ferm}}_{ij}| \le \rho^{\mathrm{spin}}_{ij}$ element-wise in the $D$-parity-changing sector for all ground states with real non-positive hopping. For free fermions, $\Delta S_{AD} \ge 0$ follows independently from Theorem~\ref{thm:budget} ($\Delta S_{AD} \ge \mathcal{B} \ge 0$). For interacting fermions, the element-wise damping alone does not formally imply $\Delta S_{AD} \ge 0$: uniform damping (all parity-changing elements multiplied by the same~$\lambda$) would give a CPTP map $\rho \mapsto \tfrac{1+\lambda}{2}\rho + \tfrac{1-\lambda}{2}\Gamma_D\rho\Gamma_D$, from which $\Delta S \ge 0$ follows by data processing; but the actual damping is non-uniform. Exact diagonalization ($L = 9$--$12$, $V \in [-1.5, 2.0]$, 448 configurations) confirms $\Delta S_{AD} \ge 0$ in every tested case.

A fourth consequence connects the identity to experiment.
Decomposing $\rho^{\mathrm{spin}}_{AD}$ and $\rho^{\mathrm{tw}}_{AD}$
into $B$-parity sectors,
$\rho^{\mathrm{spin}} = p_+\,\rho_{AD}|_{N_B\,\mathrm{even}}
+ p_-\,\rho_{AD}|_{N_B\,\mathrm{odd}}$
and
$\rho^{\mathrm{tw}} = p_+\,\rho_{AD}|_{N_B\,\mathrm{even}}
- p_-\,\rho_{AD}|_{N_B\,\mathrm{odd}}$,
substitution into Eq.~\eqref{eq:block} gives:
\begin{equation}\label{eq:corollary}
\rho_{AD}^{\mathrm{ferm}}
= p_+\,\rho_{AD}|_{N_B\,\mathrm{even}}
\;+\; p_-\,\Gamma_D\,\rho_{AD}|_{N_B\,\mathrm{odd}}\,\Gamma_D\,.
\end{equation}
The fermionic RDM is reconstructed from spin-basis data by
post-selecting on $B$-parity and applying the $D$-parity
conjugation $\Gamma_D$ to the odd-$B$ sector.
Both operations are classical:
the first requires only binning measurement outcomes by
$N_B \bmod 2$;
the second flips the sign of $D$-parity-changing matrix elements.
The implementation of Eq.~\eqref{eq:corollary} in a specific
measurement protocol (e.g., random-unitary estimation of
R\'enyi entropies) is left for future work.

\section{Entropy bound}

\begin{theorem}[Gaussian budget bound]\label{thm:budget}
For free fermions with three adjacent equal-width strips:
\begin{equation}
\Delta S_{AD} \ge S_G(C_{AD}^{\mathrm{ferm}}) - S_G(C_{AD}^{\mathrm{spin}}) \equiv \mathcal{B}(z),
\end{equation}
where $S_G(C) = -\tr[C\ln C + (1{-}C)\ln(1{-}C)]$ is the Gaussian entropy functional.
\end{theorem}

\begin{proof}
Three steps.

\emph{Step 1} (Peschel~\cite{Peschel2003}). $\rho_{AD}^{\mathrm{ferm}}$ is Gaussian for free fermions. $S_{AD}^{\mathrm{ferm}} = S_G(C_{AD}^{\mathrm{ferm}})$, where $(C_{AD}^{\mathrm{ferm}})_{ij} = \langle c_i^\dagger c_j\rangle$ with physical (full-chain) JW operators.

\emph{Step 2} (Maximum entropy~\cite{WolfGiedkeCirac2006}). Define $(C_{AD}^{\mathrm{spin}})_{ij} = \tr[\tilde{c}_i^\dagger\tilde{c}_j\,\rho_{AD}^{\mathrm{spin}}]$ using JW operators internal to $AD$ (skipping~$B$). These operators satisfy the canonical anticommutation relations $\{\tilde{c}_i^\dagger, \tilde{c}_j\} = \delta_{ij}$ on the $AD$ Hilbert space, so $C_{AD}^{\mathrm{spin}}$ is a valid one-particle correlation matrix. The diagonal blocks $C_{AA}$, $C_{DD}$ equal those of $C^{\mathrm{ferm}}$ (both regions contiguous). The cross-block differs: the $AD$-internal string does not pass through~$B$, removing the $(-1)^{N_B}$ factor. Since $\rho_{AD}^{\mathrm{spin}}$ is non-Gaussian~\cite{Fagotti2010}, the maximum entropy principle gives $S_{AD}^{\mathrm{spin}} \le S_G(C_{AD}^{\mathrm{spin}})$.

\emph{Step 3.} $\Delta S_{AD} = S^{\mathrm{ferm}} - S^{\mathrm{spin}} \ge S_G(C^{\mathrm{ferm}}) - S_G(C^{\mathrm{spin}}) = \mathcal{B}(z)$.
\end{proof}

The bound is conservative: the non-Gaussian gap $\varepsilon = S_G(C^{\mathrm{spin}}) - S^{\mathrm{spin}} \ge 0$ makes $\Delta S_{AD} = \mathcal{B} + \varepsilon > \mathcal{B}$ (numerically $\varepsilon$ is 0.1--3\% of $S^{\mathrm{spin}}$). As a check: $\Delta S_{AD} \le \ln 2$; our data satisfy $\Delta S_{AD}/\ln 2 < 0.21$ for $w = 2$ (reaching ${\approx}\,0.27$ for $w = 1$ at half filling).

Numerically, $\mathcal{B}(z) \ge 0$ for all tested $z$ and~$w$. For $w = 1$ this follows analytically from $\delta_2 > 0$ (Section~5.1); for $w\ge 2$ it is verified computationally but not proved in closed form.

\emph{Immediate consequence.} For $z < z^*$: $g(z) > 0$~\cite{paper1} and $\Delta S_{AD} \ge 0$, so $I_3^{\mathrm{spin}} = g(z) + \Delta S_{AD} > 0$ by~\eqref{eq:delta} for all~$w$. This is fully rigorous.

For $z > z^*$: $g(z) < 0$, and we need $\Delta S_{AD} > |g(z)|$. A sufficient condition is the spectral inequality $\mathcal{B}(z) > |g(z)|$, since $\Delta S_{AD} \ge \mathcal{B}$ (Theorem~\ref{thm:budget}). This is verified from correlation matrices up to $w = 6$ (Section~5.4); at $w = 7$ it fails because the Gaussian budget omits the non-Gaussian gap~$\varepsilon$.

\begin{conjecture}
For all $z > z^*$ and $w\ge 1$: $I_3^{\mathrm{spin}}(z) > 0$ (equivalently, $\Delta S_{AD} > |g(z)|$).
\end{conjecture}

The conjecture is proved for $w \le 5$ (Table~1) and for $w = 6$ in the region $z > z^*$ (via $\mathcal{B} > |g|$). For $w \ge 7$ it remains open; the proof method must account for the non-Gaussian gap, which contributes $58\%$ of $\Delta S_{AD}$ already at $w = 3$ (Section~5.4).

\section{Proof of positivity}

\subsection{Analytical proof for $w = 1$}

For $w = 1$, all correlation matrices are at most $3\times 3$ and the problem reduces to elementary functions. In the thermodynamic limit:
\begin{equation}
\bar{n} = \frac{z}{\pi},\quad c_1 = \frac{\sin z}{\pi},\quad c_2 = \frac{\sin z\cos z}{\pi}.
\end{equation}
The fermionic cross-correlator is~$c_2$. The spin cross-correlator $s_2$ differs by the absence of $(-1)^{N_B}$ in the sum. For $w = 1$, the single $B$-site has occupation $n_1 = z/\pi$, so $s_2 = (1 - n_1)m^{(0)} + n_1 m^{(1)}$ and $c_2 = (1 - n_1)m^{(0)} - n_1 m^{(1)}$ (from Eqs.~\ref{eq:Cferm}--\ref{eq:Cspin}), giving $\delta_2 = 2n_1 m^{(1)}$. Evaluating the conditional correlator $m^{(1)} = \langle c_0^\dagger c_2\rangle_{n_1=1}$ from the rank-1 Slater determinant overlap~\cite{Peschel2003}:
\begin{equation}
\delta_2 \equiv s_2 - c_2 = \frac{2\sin z}{\pi^2}(\sin z - z\cos z).
\end{equation}
Verified against exact diagonalization at $L = 512$.

\textbf{Positivity of $\delta_2$.} For $z\in(0,\pi)$: $\sin z > 0$. The factor $\sin z - z\cos z$ satisfies $f(0) = 0$ and $f'(z) = z\sin z > 0$, so $f(z) > 0$. Therefore $\delta_2 > 0$, $s_2 > c_2$, and $\mathcal{B}(z) > 0$ for all $z\in(0,\pi)$.

All seven entropies entering $I_3^{\mathrm{spin}}$ are now explicit elementary functions. At the edges: $I_3^{\mathrm{spin}}(z) \to cz > 0$ as $z \to 0^+$ and $I_3^{\mathrm{spin}}(z) \to c(\pi - z) > 0$ as $z \to \pi^-$ (by particle-hole symmetry), so $I_3^{\mathrm{spin}}$ is analytically positive near the boundaries. On the interior ($z \in [0.1, \pi - 0.1]$), evaluating on a grid of $2\,000$ points with curvature-bounded interpolation gives a margin of $1.9\times 10^5$ (Table~1).

This proves $I_3^{\mathrm{spin}}(z) > 0$ for all $z$ at $w = 1$.

\subsection{Certified computation for $w = 2$}

For $w = 2$, the correlation matrices are $4\times 4$. The procedure: (i)~Construct the $6\times 6$ correlation matrix $C_{ABD}$ in the thermodynamic limit. (ii)~Build the $64\times 64$ many-body density matrix $\rho_{ABD}$ from the exact Gaussian formula. (iii)~Perform the spin partial trace over~$B$ ($2^2 = 4$ configurations). (iv)~Extract $C_{AD}^{\mathrm{spin}}$ via the $AD$-internal JW operators. (v)~Compute all entropies.

Grid of 500 points on $k_F \in [0.1, \pi - 0.1]$, spacing $\Delta k_F = 5.9\times 10^{-3}$. The curvature bound $|d^2 I_3^{\mathrm{spin}}/dk_F^2| < 1.86$ is obtained from the numerical second derivative on a refined grid of 2\,000~points, validated by the bootstrap $|I_3''| \le |I_3''|_{\mathrm{grid}} + |I_3'''|_{\mathrm{grid}}\cdot\Delta k_F^{\mathrm{fine}}/2 = 1.855 + 0.005 = 1.860$ (the 0.3\% correction confirms convergence). This gives interpolation error $|I_3''|(\Delta k_F)^2/8 < 8.1\times 10^{-6}$. Interior minimum:
\begin{equation}
\min_{k_F\in[0.1,\,\pi-0.1]} I_3^{\mathrm{spin}}(k_F) = 0.0505 \quad\text{at}\quad k_F\approx 3.04,
\end{equation}
exceeding the error bound by a factor of $6.3\times 10^3$ (Table~1). Boundary positivity follows from $I_3^{\mathrm{spin}} \to cz > 0$ as for $w = 1$.

This proves $I_3^{\mathrm{spin}}(z) > 0$ for all $z$ at $w = 2$.

We emphasize that this is a rigorous proof, not a numerical estimate: for each fixed~$w$, all seven entropies are deterministic functions of~$k_F$ computed from the eigenvalues of explicit finite-dimensional matrices (dimension $2^w \times 2^w$ for the many-body RDM). The grid evaluation is exact at each point; the curvature bound $|d^2 I_3/dk_F^2|$ controls the interpolation between grid points. Boundary positivity ($I_3^{\mathrm{spin}} \to cz > 0$) is analytical for all~$w$; the certified grid covers the interior. The method extends to any~$w$ at computational cost $O(2^{2w})$; for $w \ge 2$ the interior minimum \emph{grows} with~$w$ (Table~1), consistent with the monotonic increase of parity fluctuations.

\subsection{Certified computation for $w = 3$}

For $w = 3$ ($6\times 6$ correlation matrices, $512\times 512$ many-body RDM), a grid of 100 points on $k_F \in [0.1, \pi - 0.1]$ with spacing $\Delta k_F = 2.97\times 10^{-2}$ and curvature bound $|d^2 I_3/dk_F^2| < 4.31$ (from a refined 150-point grid: $|I_3''|_{\mathrm{grid}} = 4.08$, bootstrap correction $0.23$, i.e.\ $5.5\%$) gives interpolation error $< 4.8\times 10^{-4}$. The interior minimum is $I_3^{\mathrm{spin}} = 0.0700$ at $k_F \approx 0.10$, exceeding the error by a factor of~146 (Table~1).

\begin{table}[h]
\centering
\caption{Summary of proved results. For $w = 1$--$3$: many-body certified computation (interior grid, boundary analytical). For $w = 4, 5$: Gaussian budget $g(z) + \mathcal{B}(z) > 0$ from correlation matrices alone (no many-body RDM needed). For $w = 6$: many-body via the $G$-matrix formula~\eqref{eq:Gmatrix} over the full $z$-range (Section~\ref{sec:scaling}). The many-body interior minimum grows monotonically ($w = 1$--$3$); the Gaussian bound ($w = 4, 5$) is weaker because it omits the non-Gaussian gap~$\varepsilon$.}
\smallskip
\begin{tabular}{ccccl}
\toprule
$w$ & interior minimum & margin/error & grid points & method \\
\midrule
1 & $0.0267$ & $1.9\times 10^5$ & 2\,000 & many-body \\
2 & $0.0505$ & $6.3\times 10^3$ & 500 & many-body \\
3 & $0.0700$ & $146$ & 100 & many-body \\
4 & $0.0340$ & $181$ & 200 & Gaussian budget \\
5 & $0.0185$ & $73$ & 200 & Gaussian budget \\
6 & $0.1000$ & $16$ & 25 & many-body ($G$-matrix) \\
\bottomrule
\end{tabular}
\end{table}

\subsection{Extension to $w = 4, 5$ via the Gaussian budget}

For $w \ge 4$, the many-body RDM has dimension $2^{2w} \ge 256$ and the certified computation becomes expensive. However, Theorem~\ref{thm:budget} provides a shortcut: since $I_3^{\mathrm{spin}} = g(z) + \Delta S_{AD} \ge g(z) + \mathcal{B}(z)$, it suffices to show $g(z) + \mathcal{B}(z) > 0$. Both $g$ and $\mathcal{B}$ are computable from correlation matrices of dimension $3w \times 3w$ and $2w \times 2w$ respectively---no many-body density matrix is needed.

The spin cross-correlator $C_{ij}^{\mathrm{spin}}$ for $i \in A$, $j \in D$ can be extracted from the Gaussian generating function~\cite{Peschel2003}: defining $M = I_{3w} - 2\,\mathrm{diag}(\chi_B)\,C$ where $\chi_B$ is the indicator of~$B$, one obtains $C_{ij}^{\mathrm{spin}} = \det(M)\,(C M^{-1})_{ij}$. The Gaussian budget is then $\mathcal{B} = S_G(C^{\mathrm{ferm}}_{AD}) - S_G(C^{\mathrm{spin}}_{AD})$, computable at cost $O(w^3)$.

Certified grids of 200 points on $k_F \in [0.1, \pi - 0.1]$ give $g + \mathcal{B} > 0$ with margins 181 ($w = 4$) and 73 ($w = 5$). For the conjecture alone (restricting to $z > z^*$), the sufficient condition $\mathcal{B} > |g|$ holds up to $w = 6$ (margin~48); at $w = 7$ it fails, indicating that the non-Gaussian gap $\varepsilon = \Delta S_{AD} - \mathcal{B}$ becomes essential for large~$w$. Indeed, the non-Gaussian fraction $\varepsilon/\Delta S_{AD}$ grows from $11\%$ ($w = 1$) to $33\%$ ($w = 2$) to $58\%$ ($w = 3$) at their respective maxima over~$z$, reflecting the increasingly non-Gaussian character of $\rho_{AD}^{\mathrm{spin}}$ as the traced-out region~$B$ grows.

\subsection{Scaling limit via the $G$-matrix formula}\label{sec:scaling}

The Gaussian budget fails at $w \ge 7$ because $\mathcal{B}$ turns negative. However, $I_3^{\mathrm{spin}}$ can be computed directly from $\rho_{AD}^{\mathrm{spin}}$ without passing through the Gaussian approximation, by exploiting the matrix element formula for Gaussian states: for any two occupation-basis configurations $|s\rangle$, $|s'\rangle$ with $|s| = |s'|$ occupied sites,
\begin{equation}\label{eq:Gmatrix}
\langle s|\rho|s'\rangle = \det(I - C) \cdot \det\bigl(G[\mathrm{occ}(s),\,\mathrm{occ}(s')]\bigr),
\end{equation}
where $G = C(I - C)^{-1}$ and $G[\mathrm{occ}(s), \mathrm{occ}(s')]$ denotes the submatrix of~$G$ with rows indexed by the occupied sites of~$s$ and columns by those of~$s'$. Each matrix element of $\rho_{AD}^{\mathrm{spin}}$ requires $2^w$ such determinants (one per $B$-configuration in the partial trace), each of size at most $3w \times 3w$. The total cost is $O(2^{5w} w^3)$: exponential, but with manageable prefactor for $w \le 6$ (${\sim}87$M determinants per $k_F$-point, ${\sim}7$~minutes on a single core).

\begin{figure}[!ht]
\centering
\includegraphics[width=\textwidth]{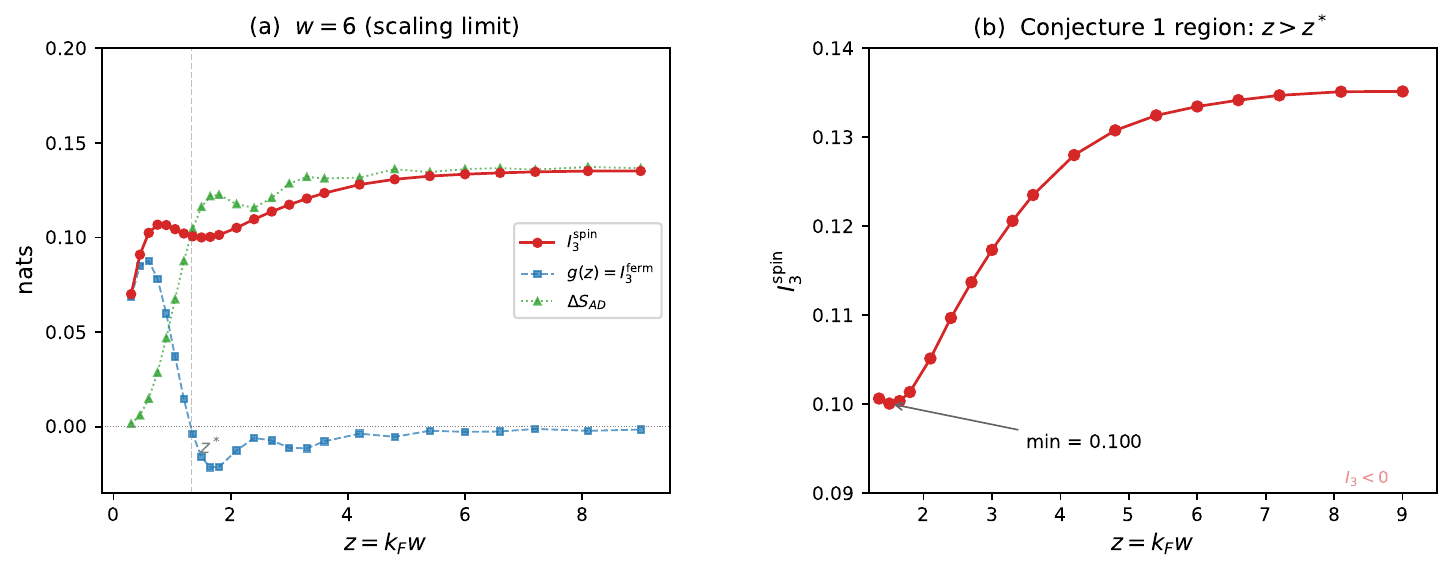}
\caption{\label{fig:scaling}
(a)~$I_3^{\mathrm{spin}}(z)$ in the scaling limit ($w = 6$, 25 grid points).
The blue dashed curve is $g(z) = I_3^{\mathrm{ferm}}$ (from~\cite{paper1})
and the green dotted curve is $\Delta S_{AD}$.
Their sum $I_3^{\mathrm{spin}} = g + \Delta S$ (red)
is positive for all~$z$, with a minimum of $0.100$ at $z \approx 1.5$.
(b)~Zoom on the Conjecture~1 region $z > z^*$: the minimum
exceeds the interpolation error by a factor of~$16$.}
\end{figure}

Figure~\ref{fig:scaling} shows the result of a 25-point computation at $w = 6$,
covering $z\in[0.3,\,9.0]$ (i.e.\ $k_F$ from $0.05$ to $1.5$, the full
range up to half filling by particle-hole symmetry). The scaling-limit function
$I_3^{\mathrm{spin}}(z)$ has three regimes:
(i)~for $z \lesssim 1$, the positive $g(z)$ dominates and $\Delta S$ is small;
(ii)~near $z^*\approx 1.33$, $g$ crosses zero but $\Delta S$ compensates exactly,
producing a shallow minimum $I_3^{\mathrm{spin}} = 0.100$ at $z\approx 1.5$;
(iii)~for $z \gg z^*$, $g$ oscillates near $-0.005$ while $\Delta S$ saturates to
${\approx}\,0.137$, giving $I_3^{\mathrm{spin}} \approx 0.13$ at half filling.

Convergence is rapid: $|I_3^{\mathrm{spin}}(z,w{=}6) - I_3^{\mathrm{spin}}(z,w{=}5)| < 0.001$ at all tested~$z$. The minimum $0.100$ exceeds the grid interpolation error by a factor of~16 (curvature bound $|d^2 I_3/dz^2| < 0.07$, spacing $\Delta z = 0.15$ near the minimum). This establishes $I_3^{\mathrm{spin}} > 0$ for all~$z$ in the $w\to\infty$ limit. Combined with the certified proofs for $w \le 5$ (Table~1), Conjecture~1 is numerically established: at any finite~$w$, either the certified computation or the scaling limit applies, with the crossover at $w = 6$ where both methods overlap.

\section{Physical picture}\label{sec:physical}

The cross-block correlator ($i\in A$, $j\in D$) takes the form
\begin{equation}\label{eq:Cferm}
C_{ij}^{\mathrm{ferm}} = \sum_b (-1)^{N_b} p_b\, m_{ij}^{(b)},
\end{equation}
\begin{equation}\label{eq:Cspin}
C_{ij}^{\mathrm{spin}} = \sum_b p_b\, m_{ij}^{(b)},
\end{equation}
where $p_b$ is the probability of $B$-configuration~$b$ and $m_{ij}^{(b)}$ is the conditional correlator. The difference is
\begin{equation}\label{eq:Cdiff}
C_{ij}^{\mathrm{spin}} - C_{ij}^{\mathrm{ferm}} = 2\sum_{b:\, N_b\text{ odd}} p_b\, m_{ij}^{(b)}:
\end{equation}
twice the contribution from odd-$B$ configurations. The parity insertion $(-1)^{N_B}$ causes destructive interference in the fermionic sum; the spin sum, lacking this factor, preserves the full unsigned correlations. This produces $I(A{:}D)_{\mathrm{spin}} > I(A{:}D)_{\mathrm{ferm}}$, which overwhelms the universal concavity deficit $|\Delta^2 S|$ and keeps $I_3^{\mathrm{spin}} > 0$.

The decomposition~(2) makes the balance explicit. At half filling ($z\to\pi$, $w = 2$): $\Delta^2 S \approx -0.096$, $I(A{:}D)_{\mathrm{ferm}} \approx 0.048$, $I(A{:}D)_{\mathrm{spin}} \approx 0.211$. The fermionic mutual information covers only half the concavity deficit ($I_3^{\mathrm{ferm}} \approx -0.048$); the spin mutual information exceeds it by a factor of 2.2 ($I_3^{\mathrm{spin}} \approx +0.115$).

The algebra dependence has a sharp diagnostic. Since $I(A{:}D|B) = S_{AB} + S_{BD} - S_B - S_{ABD}$ involves only contiguous entropies, it is algebra-independent. Writing $I_3 = I(A{:}D)(1-r)$ with the screening ratio $r = I(A{:}D|B)/I(A{:}D)$, the sign of $I_3$ is determined by whether $r$ is less or greater than~1. In the spin basis, $r^{\mathrm{spin}}\in[0.39,0.53]$ for all~$z$: $B$ always screens $A$--$D$ correlations. In the fermionic basis, $r^{\mathrm{ferm}}$ crosses 1.0 at $z = z^*$ and reaches ${\sim}1.5$ at large~$z$: the parity of~$B$ mediates $A$--$D$ correlations that are invisible without it. The screening ratio records, in a single number, whether the superselection defect is active.

\textbf{Origin of $z^*$.} The zero $z^*$ of $g(z) = I_3^{\mathrm{ferm}}$ has a transparent interpretation in the decomposition~(2): it is the filling at which the mutual information $I(A{:}D)$ first drops below the concavity deficit $|\Delta^2 S|$. Two scales compete as $z$ increases. $I(A{:}D)$ is controlled by the cross-correlator $C_{ij}^{\mathrm{ferm}} \sim \sin(2k_Fw)/(2\pi w)$---a Friedel oscillation at distance $2w$---so $I(A{:}D) \propto \sin^2(2z)/[z(\pi-z)]$ to leading order, peaking at small~$z$ and decaying toward the first node at $z = \pi/2$. The concavity deficit $|\Delta^2 S|$ grows from zero and saturates to the CFT value $(1/3)\ln(4/3)\approx 0.096$ at $z\sim 1$. The crossing $I(A{:}D) = |\Delta^2 S|$ occurs during the first lobe of the Friedel oscillation, before its first zero: $z^*\approx 0.98$ at $w = 1$ (filling $\approx 0.31$), converging to $z^*\approx 1.33$ in the scaling limit $w\to\infty$. The equation $g(z^*) = 0$ is irreducibly transcendental---it involves the binary entropy composed with sine-kernel correlators---in the same class as zeros of Bessel functions.

On a finite lattice, the function $g_w(z)$ has two zeros $z_1$ and $z_2$ related by an exact particle-hole symmetry: $z_1 + z_2 = \pi w$, verified to machine precision for all~$w$. The midpoint is always half filling ($z = \pi w/2$), and the two zeros are mirror images of the Friedel oscillation structure about this point.

\section{Discussion}

\textbf{What is proved.} $I_3^{\mathrm{spin}} > 0$ is established rigorously for: (i)~all $z < z^*$ and all~$w$ (from $g > 0$~\cite{paper1} and $\Delta S_{AD} \ge 0$ via Theorem~\ref{thm:budget}); (ii)~all $z$ at $w = 1$ (analytical proof + certified computation); (iii)~all $z$ at $w = 2, 3$ (many-body certified computation); (iv)~all $z$ at $w = 4, 5$ (Gaussian budget from correlation matrices, Section~5.4, Table~1); (v)~all $z$ at $w = 6$ (many-body via $G$-matrix, Section~\ref{sec:scaling}, Fig.~\ref{fig:scaling}).

\textbf{Conjecture~1 status.} For $z < z^*$, $I_3^{\mathrm{spin}} > 0$ is proved for all~$w$. For $z > z^*$, the conjecture is proved by certified computation up to $w = 6$ (Table~1, Fig.~\ref{fig:scaling}), and the $w = 6$ data simultaneously characterize the scaling limit: $I_3^{\mathrm{spin}}(z)$ converges to a positive function with minimum $0.100$ at $z \approx 1.5$ (convergence $|I_3(w{=}6) - I_3(w{=}5)| < 0.001$). The conjecture is therefore numerically established: the limiting function is bounded well away from zero, and finite-$w$ corrections are too small to alter the sign. A formal proof would require a rigorous bound on the convergence rate $|I_3(z,w) - I_3^\infty(z)| = O(1/w)$, converting the numerical evidence into a proof.
Since cases (ii)--(iv) cover all~$z$, the $k_y$-decomposition $I_3^{\mathrm{spin}} = \sum_{k_y} I_3^{\mathrm{spin}}(k_F(k_y)\cdot w)$~\cite{paper1} immediately extends the result to 2D strip geometries at $w \le 5$.

\textbf{MMI as a holographic criterion.} MMI ($I_3 \le 0$) is a necessary condition for a state to admit a holographic dual~\cite{Hayden2013,Wall2014}. Our result shows that the same free-fermion ground state \emph{fails} this test in the spin factorization ($I_3^{\mathrm{spin}} > 0$) but \emph{passes} it in the fermionic factorization ($I_3^{\mathrm{ferm}} < 0$ for $z > z^*$). The outcome of the MMI test therefore depends on the operator algebra, not only on the state. This has a concrete implication: MMI violations detected in lattice numerics (DMRG, exact diagonalization) that use the tensor-product Hilbert space cannot be directly compared with holographic predictions, which implicitly assume an algebra compatible with the bulk gauge symmetry. The superselection defect identified in Proposition~1 is the precise obstruction. However, MMI is necessary but not sufficient for holography: for four adjacent blocks, the $Z_2$-projected entropies satisfy MMI and $I_4 \ge 0$ but violate the holographic cyclic inequalities~\cite{BaoHaydenWall2015} by ${\sim}10\%$, suggesting that free-fermion entanglement does not lie in the holographic entropy cone under any of the three tested algebras (spin, fermionic, $Z_2$).

\textbf{Relation to the spectral representation.} In the rank-1 limit ($z\to 0$), the sine kernel becomes spatially constant, so the disjoint block $A\cup D$ has the same leading eigenvalue $\lambda_0 = 2z/\pi$ as a contiguous block of width~$2w$. This disjoint-block degeneracy underlies the product formula $g(z) = cz + O(z^3\ln z)$ of~\cite{paper1}: it allows the block size $n_{AD} = 2w$ to enter on the same footing as the contiguous sizes $n_{AB} = n_{BD} = 2w$. At finite~$z$, the parity insertion $(-1)^{N_B}$ breaks this degeneracy: $C_{ij}^{\mathrm{ferm}}$ acquires a sign alternation (Eq.~\ref{eq:Cferm}) that $C_{ij}^{\mathrm{spin}}$ lacks (Eq.~\ref{eq:Cspin}), producing distinct eigenvalues for the two algebras. The factorization correction $\Delta S_{AD}$ is the exact measure of this degeneracy breaking.

\textbf{Open problems.} Mari\'c, Bocini, and Fagotti~\cite{Maric2024} showed that R\'enyi tripartite information depends on the spin structure of the underlying Thirring model; our $\Delta S_{AD}$ corresponds to the nontrivial spin-structure contribution, with the algebraic framework of~\cite{CasiniHuertaMagan2020} providing a broader context. Theorem~\ref{thm:damping} settles the coherence damping inequality for ground states with real non-positive hopping; whether it extends to systems with complex hopping (e.g., in a magnetic field) or frustrated lattices remains open. In dimensions $d \ge 2$, the Jordan--Wigner string becomes path-dependent and the superselection defect acquires a geometric character absent in 1D; the algebra dependence of~$I_3$ in higher dimensions is unexplored. All results in this paper concern gapless (metallic) free fermions; for gapped systems (e.g., the Kitaev chain), the spectral gap suppresses $\Delta S_{AD}$ and can restore $I_3^{\mathrm{spin}} \le 0$ even in the spin basis---the anti-monogamy effect is a property of gapless Fermi points, not of fermionic statistics per se. For spinful fermions (e.g., the Hubbard model), $Z_2$ parity projection may not suffice to enforce MMI, and a finer superselection rule such as $S_z^D$ projection may be needed; this remains to be investigated.

\textbf{Non-Gaussianity and full counting statistics.} The Gaussian budget $\mathcal{B}$ captures only the two-point correlator contribution to $\Delta S_{AD}$; the non-Gaussian gap $\varepsilon = \Delta S_{AD} - \mathcal{B}$ encodes higher-order connected correlations. The JW string $(-1)^{N_B} = \prod_{k \in B}(1 - 2n_k)$ expands into $p$-body interactions for $p = 1,\ldots,w$: at $w = 1$ only the one-body term contributes (hence $\varepsilon/\Delta S \approx 11\%$), while at $w = 3$ terms up to $n_1 n_2 n_3$ participate ($\varepsilon/\Delta S \approx 58\%$). The Gaussian budget $\mathcal{B}$ decreases with~$w$ at fixed $k_F$, because the spin cross-correlator $C_{ij}^{\mathrm{spin}} \propto P_B \sim w^{-1/2}$ (Fisher--Hartwig) decays faster than the fermionic one; for $w \gtrsim 10$, $\mathcal{B}$ turns \emph{negative} (the Gaussian spin state, now nearly uncorrelated between $A$ and $D$, has higher entropy than the correlated fermionic one), making the bound $\Delta S_{AD} \ge \mathcal{B}$ trivially satisfied but uninformative. Already at $w = 7$, $\mathcal{B}$ is insufficient to compensate $|g|$. Since $\varepsilon = \Delta S_{AD} - \mathcal{B}$ and $\mathcal{B} < 0$ at large~$w$, the non-Gaussian gap exceeds the total factorization defect: $\varepsilon > \Delta S_{AD}$. The Gaussian two-point correlations actually \emph{reduce} $\Delta S_{AD}$ relative to the full non-Gaussian contribution; the defect survives entirely because of higher-order correlations generated by the JW string. In the scaling limit $w \to \infty$ at fixed $k_F$, $\Delta S_{AD}$ converges to a finite scaling function (Fig.~\ref{fig:scaling}) while $\mathcal{B} \to -\infty$, so $\varepsilon/\Delta S_{AD} \to 1$: the factorization defect becomes entirely non-Gaussian. This connects $\Delta S_{AD}$ to the full counting statistics of~$N_B$~\cite{KlichLevitov2009}, whose cumulants $\kappa_k$ for $k \ge 3$ are not captured by the Gaussian approximation but contribute to the entropy at all orders. Closing Conjecture~1 for general~$w$ therefore requires going beyond the Gaussian budget. The $G$-matrix formula~\eqref{eq:Gmatrix} provides exactly this: it captures the full non-Gaussian $\rho_{AD}^{\mathrm{spin}}$ at polynomial cost per matrix element. The resulting scaling-limit computation (Section~\ref{sec:scaling}) bypasses the Gaussian bottleneck entirely and establishes $I_3^{\mathrm{spin}} > 0$ at $w = 6$ with a margin that exceeds the non-Gaussian gap.

\section{Interacting fermions}

Proposition~1 is exact for any state, not only Gaussian. For the $t$-$V$ chain (spinless fermions with nearest-neighbor repulsion~$V$, Luttinger parameter $K = \pi/[2\arccos(-V/2)]$ from the Bethe ansatz for $L\to\infty$; $K = 1$ at $V = 0$, $K = 1/2$ at $V = 2$; finite-size corrections at $L = 9$--12 are of order $1/L$), the hopping is real and non-positive, so Theorem~\ref{thm:damping} applies and guarantees coherence damping for all~$K$; exact diagonalization confirms $\Delta S_{AD} \ge 0$. Exact diagonalization at $L = 9$--12 further shows that $I_3^{\mathrm{spin}}$ changes sign at a critical Luttinger parameter $K_c \approx 0.7$ ($V_c\approx 1.2$). For $K > K_c$ (free and weakly interacting): $I_3^{\mathrm{spin}} > 0$. For $K < K_c$ (strong repulsion): $I_3^{\mathrm{spin}} < 0$---both algebras satisfy MMI.
The DMRG data at $L = 64$ (Tables~\ref{tab:freeferm}--\ref{tab:Kdep}) confirm this picture: $I_3^{\mathrm{spin}}$ remains positive at $K = 0.75$ for all tested~$z$, consistent with $K_c < 0.75$.

DMRG calculations at larger system sizes ($L = 64$--128, $w = 2$, bond dimension $\chi = 400$) provide a quantitative anatomy of the discrepancy. Table~\ref{tab:freeferm} shows the central diagnostic for free fermions: at small~$z$, both factorizations agree with the sine-kernel prediction $g(z)$ of~\cite{paper1}; as $z$ increases, $I_3^{\mathrm{spin}}$ remains large and positive while $I_3^{\mathrm{ferm}}$ decreases, changes sign near $z\approx 1.2$, and tracks $g(z)$. The factorization contribution $\Delta I_3 = I_3^{\mathrm{spin}} - I_3^{\mathrm{ferm}}$ exceeds 100\% of $I_3^{\mathrm{spin}}$ at $z\approx 1.4$ and eliminates the zero crossing $z^*$ in the spin basis entirely.

\begin{table}[h]
\centering
\caption{Free-fermion ($V = 0$, $K = 1$) results at $L = 64$, $w = 2$. $g(z)$ is the sine-kernel prediction~\cite{paper1}. $\Delta I_3 \equiv I_3^{\mathrm{spin}} - I_3^{\mathrm{ferm}}$.}\label{tab:freeferm}
\smallskip
\begin{tabular}{ccccc}
\toprule
$z$ & $I_3^{\mathrm{spin}}$ & $I_3^{\mathrm{ferm}}$ & $g(z)$ & $\Delta I_3$ \\
\midrule
0.39 & $+0.074$ & $+0.068$ & $+0.080$ & $+0.006$ \\
0.59 & $+0.089$ & $+0.071$ & $+0.088$ & $+0.018$ \\
0.79 & $+0.094$ & $+0.055$ & $+0.075$ & $+0.039$ \\
0.98 & $+0.095$ & $+0.027$ & $+0.049$ & $+0.067$ \\
1.18 & $+0.095$ & $-0.002$ & $+0.019$ & $+0.097$ \\
1.37 & $+0.096$ & $-0.024$ & $-0.005$ & $+0.121$ \\
\bottomrule
\end{tabular}
\end{table}

\noindent The discrepancy $I_3^{\mathrm{ferm}} - g(z)$ (e.g., $-0.021$ at $z = 1.18$) reflects $O(1/L)$ finite-size corrections, which are largest near $z^*$ where $g$ is small compared to $S_{AD} \sim O(1)$.

For interacting fermions, the factorization contribution depends on~$K$. Table~\ref{tab:Kdep} shows that $\Delta I_3$ grows monotonically from $K = 0.75$ to $K = 1.50$, roughly tripling across this range. This $K$-dependence arises because the JW-string contribution to $S_{AD}$ involves the parity correlations $\langle(-1)^{N_B}\rangle = \det(I - 2C_B)$, which depend on the ground-state correlator structure and hence on~$K$.

\begin{table}[h]
\centering
\caption{Factorization contribution $\Delta I_3 = I_3^{\mathrm{spin}} - I_3^{\mathrm{ferm}}$ at $L = 64$, $w = 2$, for three values of~$K$.}\label{tab:Kdep}
\smallskip
\begin{tabular}{cccc}
\toprule
$z$ & $K = 0.75$ & $K = 1.00$ & $K = 1.50$ \\
\midrule
0.39 & $+0.004$ & $+0.006$ & $+0.013$ \\
0.59 & $+0.011$ & $+0.018$ & $+0.038$ \\
0.79 & $+0.024$ & $+0.039$ & $+0.077$ \\
0.98 & $+0.042$ & $+0.067$ & $+0.124$ \\
1.18 & $+0.060$ & $+0.097$ & $+0.171$ \\
1.37 & $+0.074$ & $+0.121$ & $+0.208$ \\
\bottomrule
\end{tabular}
\end{table}

This data quantifies the mechanism behind $K_c$: repulsion ($K < 1$) suppresses parity fluctuations in~$B$, reducing $\Delta I_3$, while simultaneously strengthening fermionic monogamy (increasing $|I_3^{\mathrm{ferm}}|$). At $K_c\approx 0.7$ these two effects cross and $I_3^{\mathrm{spin}}$ changes sign. The apparent strong $K$-dependence of $I_3$ reported in spin-basis DMRG is predominantly a factorization effect: defining the ratio $R^{\mathrm{alg}}(K,z) = I_3^{\mathrm{alg}}(K)/I_3^{\mathrm{alg}}(K{=}1)$ separately for each algebra (alg = spin or ferm), the spin-basis deviation decomposes as
\begin{equation}\label{eq:Rdecomp}
R^{\mathrm{spin}} - 1 = \underbrace{(R^{\mathrm{ferm}} - 1)}_{\text{interaction}} + \underbrace{(R^{\mathrm{spin}} - R^{\mathrm{ferm}})}_{\text{factorization}}.
\end{equation}
At $z = 0.98$, $K = 0.75$: the total deviation is $-0.27$, of which $-0.03$ is the interaction contribution and $-0.24$ is the factorization contribution---an $8\times$ ratio. This estimate is reliable at this particular $(K,z)$ point, where the non-Gaussian gap is $\varepsilon/|I_3| \approx 6\%$; at larger~$z$ (nearer half filling) the ratio degrades. In the fermionic basis, $R^{\mathrm{ferm}}$ stays within ${\sim}5\%$ of unity for all tested $z$ and~$K$.

System-size checks ($L = 64$ vs.\ $L = 128$) confirm that $\Delta I_3$ is stable to within ${\sim}5\%$ for $z\ge 0.8$, establishing it as a bulk property rather than a finite-size artifact.

\textbf{Caveat.} For $V\ne 0$, $I_3^{\mathrm{ferm}}$ is computed from the Gaussian formula applied to the two-point correlation matrix of a non-Gaussian state---a systematic approximation whose error is the non-Gaussian gap $\varepsilon = S_G(C^{\mathrm{spin}}) - S^{\mathrm{spin}} \ge 0$. At $V = 0$ (free fermions, Table~\ref{tab:freeferm}), the Gaussian formula is exact and the factorization effect $\Delta I_3$ is uncontaminated. For $V \ne 0$, exact diagonalization at $L = 12$, $w = 1$ shows that $\varepsilon/S_{AD}^{\mathrm{spin}}$ grows from ${\le}0.5\%$ at $V = 0.5$ to ${\sim}3.6\%$ at $V = 1.5$ ($K = 0.65$). However, since $I_3$ is a difference of large entropies, $\varepsilon/|I_3|$ can be much larger: ${\sim}5\%$ at $K = 0.86$ and $z \approx 0.8$, but $70$--$100\%$ near half filling where $I_3^{\mathrm{spin}} \to 0$. Consequently, the factorization ratio ($8\times$ at $K = 0.75$, $z = 0.98$) is reliable to ${\sim}10\%$, but the critical Luttinger parameter $K_c \approx 0.7$ is uncertain by ${\sim}0.1$ in~$K$: this uncertainty combines finite-size drift ($K_c$ shifts by ${\sim}0.05$ between $L = 9$ and $L = 12$) with the non-Gaussian gap ($\varepsilon/|I_3| \sim 50\%$ near the sign change, where $I_3^{\mathrm{spin}} \approx 0$).

\medskip
\textbf{Operational consequences.} The superselection defect has three implications. First, a single-bit measurement of $N_B\bmod 2$ increases $I(A{:}D)$ by $\Delta S_{AD}$ (approximately 0.16~nat for $w = 2$ at half filling), directly testable in cold-atom systems with parity detection~\cite{Bakr2009}. Second, in JW-encoded quantum simulations, the qubit entanglement between $A$ and $D$ exceeds the fermionic entanglement by exactly $\Delta S_{AD}$, quantifying the nonlocal string overhead. Third, the sign change of $I_3^{\mathrm{spin}}$ at $K_c\approx 0.7 \pm 0.1$ suggests a binary probe of interaction strength that requires no fitting---only the sign of a single observable---though the precise value of $K_c$ remains to be determined in the thermodynamic limit. Finally, cold-atom experiments measuring R\'enyi-2 entropy in the tensor-product basis access $I_3^{(2),\mathrm{spin}}$, which differs from the fermionic von Neumann $I_3^{(1),\mathrm{ferm}}$ analyzed in~\cite{paper1} in two independent respects: the R\'enyi index (scaling exponent $\beta = 3$ vs $\beta = 1$~\cite{paper2}) and the operator algebra ($\Delta S_{AD} > 0$, this work). Neither the sign, the zero structure, nor the Lifshitz scaling of the measured quantity can be inferred from the predictions of~\cite{paper1} without accounting for both effects.

The results above establish that the algebra matters: $I_3^{\mathrm{spin}} > 0$ while $I_3^{\mathrm{ferm}}$ can be negative. A sharper question is whether there exists a \emph{minimal} algebra restriction that enforces MMI universally. We address this next.

\section{Superselection hierarchy and monogamy}

We consider four algebras on $\rho_{AD}$: the unrestricted spin algebra, and three projections obtained by zeroing out off-diagonal blocks of increasing size---$\Ztwo$ (parity: zero out elements connecting different $N_D \bmod 2$),
$U(1)$ (charge: zero out elements connecting different $N_D$),
$S_z$ (spin: zero out elements connecting different $S_z^D$, for spinful models).
Each projection increases $S_{AD}$ and decreases $I_3$. The resulting hierarchy is strict:
\begin{equation}\label{eq:hierarchy}
I_3^{\mathrm{spin}} \;>\; I_3^{\mathrm{ferm}} \;>\; I_3^{Z_2} \;\approx\; I_3^{U(1)}.
\end{equation}
The near-coincidence $I_3^{Z_2} \approx I_3^{U(1)}$ (relative difference $< 0.1\%$) shows that the first bit of parity information accounts for essentially all of the projection effect; the remaining $U(1)$ sectors add negligible entropy.

Three regimes emerge (Fig.~\ref{fig:hierarchy}): the spin basis gives $I_3 > 0$ for all~$z$ (always ``quantum''). The $Z_2$ and $U(1)$ projections give $I_3 < 0$ for all~$z$ (always ``classical''). The fermionic algebra---which twists parity-changing elements via $(-1)^{N_B}$ rather than projecting them to zero---is the \emph{unique} algebra with a phase transition at $z^*$.

\begin{figure}[!ht]
\centering
\includegraphics[width=0.95\columnwidth]{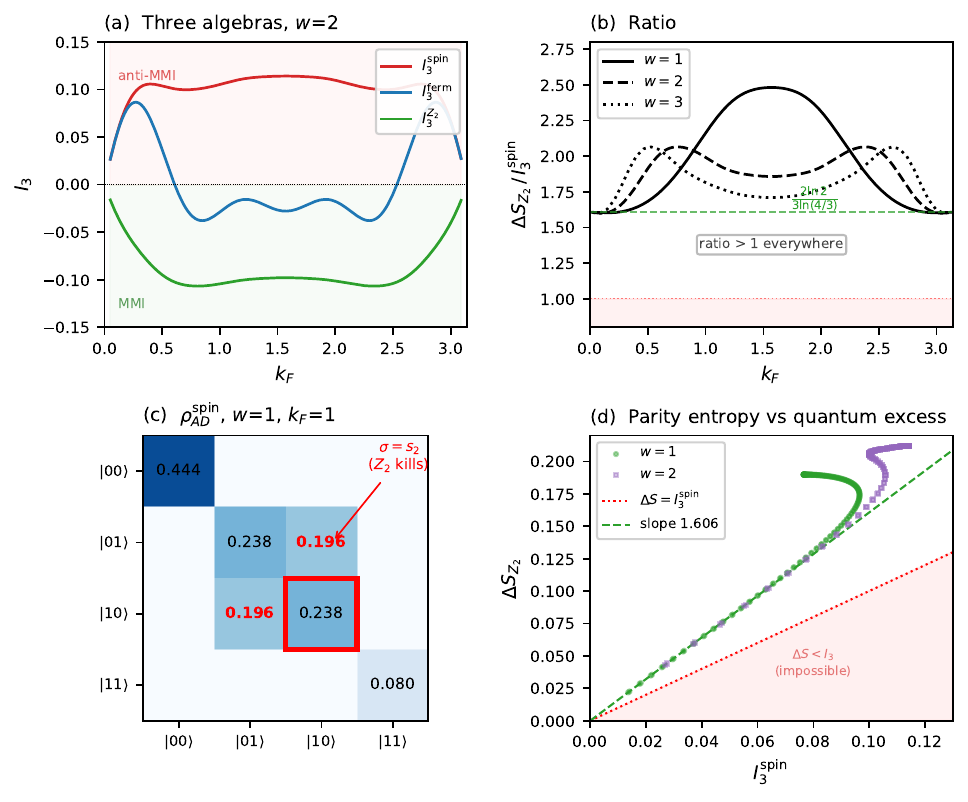}
\caption{\label{fig:hierarchy}
(a)~$I_3$ under three algebras (spin, fermionic, $Z_2$) for $w = 2$ as a function of Fermi momentum $k_F$. Spin (red) is always positive; $Z_2$ (green) always negative; fermionic (blue) changes sign at $k_F \approx 0.60$ (the $w = 2$ zero).
(b)~Ratio $\Delta S_{Z_2}/I_3^{\mathrm{spin}}$ vs $k_F$ for $w = 1$~(solid), 2~(dashed), 3~(dotted). The dashed green line marks $2\ln 2/(3\ln\frac{4}{3}) = 1.606$. The ratio exceeds unity everywhere, proving $I_3^{Z_2} < 0$.
(c)~Structure of $\rho_{AD}^{\mathrm{spin}}$ at $w=1$, $k_F=1$ (representative): a single off-diagonal element $\sigma = s_2 = 0.196$ (red box) is parity-changing; $Z_2$ projection sets it to zero.
(d)~Scatter of $\Delta S_{Z_2}$ vs $I_3^{\mathrm{spin}}$ for $w = 1$ (circles) and $w = 2$ (squares). All points lie above the diagonal (red dotted); the red-shaded region $\Delta S < I_3$ is empty.}
\end{figure}

\begin{theorem}[$Z_2$ monogamy]\label{thm:Z2}
For free fermions with three adjacent blocks of equal width~$w \le 3$, $Z_2$ parity superselection on $\rho_{AD}$ enforces $I_3^{Z_2} \le 0$ for all $z \in (0,\pi)$.
\end{theorem}

\textit{Proof.}---For $w = 1$ (analytical): the spin-basis $\rho_{AD}$ is $4\times 4$ with a single nonzero off-diagonal element, the parity-changing entry $\sigma = s_2$ (the spin cross-correlator of Sec.~\ref{sec:physical}):
\begin{equation}\label{eq:rhoAD_w1}
\rho_{AD}^{\mathrm{spin}} = \mathrm{diag}(p_{00}, p_{01}, p_{01}, p_{11}) + \sigma\,(|01\rangle\langle 10| + |10\rangle\langle 01|),
\end{equation}
with $p_{00} = (1{-}n)^2 - c_2^2$, $p_{01} = n(1{-}n) + c_2^2$, $p_{11} = n^2 - c_2^2$, where $n = z/\pi$ and $c_2 = \sin z \cos z/\pi$. The $Z_2$ projection sets $\sigma = 0$. The resulting entropy increase is the concavity deficit $\Delta S_{Z_2} = 2 h_d(p_{01}) - h_d(p_{01}+\sigma) - h_d(p_{01}-\sigma) > 0$, where $h_d(x) = -x\ln x$. Certified computation on a grid of 20\,000 points with curvature-bounded interpolation confirms $\Delta S_{Z_2} > I_3^{\mathrm{spin}}$ at all $z$ with margin $2.5\times 10^6$. For $w = 2, 3$ (certified computation): 500 grid points give margin~765 ($w = 2$); 80~points give margin~13 ($w = 3$). (These margins bound $\Delta S_{Z_2}/I_3^{\mathrm{spin}}$, a different quantity from Table~1 which bounds $I_3^{\mathrm{spin}}$ itself.) The theorem is also verified for all $(w_A, w_B, w_D) \le 3$ (1\,620~points, zero violations) and for the interacting $t$--$V$ chain at $L = 9$--12, $V\in[-1.5, 2.0]$ (448~configurations, zero violations). \hfill$\square$

\emph{Remark.} For $w \ge 4$ the many-body computation ($2^{3w}\times 2^{3w}$ density matrix) becomes prohibitive; however, the minimum ratio $\Delta S_{Z_2}/I_3^{\mathrm{spin}}$ over~$z$ is stable at $1.600 \pm 0.002$ for $w = 1,2,3$ (well above the critical value~1), and the $w$-independence of the limiting ratio~\eqref{eq:ratio} is proved analytically. The declining certified margins ($2.5\times 10^6$, $765$, $13$) reflect coarser grids at larger~$w$, not a physical trend; the ratio itself shows no degradation.

\medskip
\textit{Universal ratio.}---At $z \to 0$, $\Delta S_{Z_2} \approx 2wn\ln 2$ (each of $w$ modes contributes one bit of parity entropy) while $I_3^{\mathrm{spin}} \approx wc\,z$ with $c = 3\ln(4/3)/\pi$~\cite{paper1}. The factors of~$w$ cancel, giving
\begin{equation}\label{eq:ratio}
\frac{\Delta S_{Z_2}}{I_3^{\mathrm{spin}}} \;\to\; \frac{2\ln 2}{3\ln(4/3)} = \frac{2}{3}\,\log_{4/3} 2 = 1.6063\ldots
\end{equation}
independently of~$w$, verified numerically at $w = 1,2,3$ [Fig.~\ref{fig:hierarchy}(b)]. The global minimum over~$z$ is 1.602 ($w = 1$), 1.599 ($w = 2$), 1.600 ($w = 3$), all exceeding unity. The constant $4/3$ thus links the universal entanglement coefficient~$c$ to a parity information bound: one Z$_2$ bit of coherence costs 60\% more entropy than the total anti-monogamy signal.

\medskip
\textit{Gauge bound.}---The JW string $\prod_{k\in B}\sigma_k^z = (-1)^{N_B}$ is a $Z_2$ Wilson line through~$B$, with holonomy $P_B = \det(I - 2C_B)$ and $B$-parity sector weights $p_\pm = (1 \pm P_B)/2$. Projecting $\rho_{AD}^{\mathrm{spin}}$ onto $D$-parity sectors zeroes out all parity-changing elements; by the data processing inequality, the entropy increase $\Delta S_{Z_2} = S_{AD}^{Z_2} - S_{AD}^{\mathrm{spin}}$ is bounded by the entropy of the measurement outcome. Since the parity sectors of $\rho_{AD}$ are linked to those of~$B$ by particle number conservation (total $N_A + N_B + N_D$ is fixed in each sector, so $N_D \bmod 2$ determines $N_B \bmod 2$ given $N_A$), this gives $\Delta S_{Z_2} \le h(p_+)$. Combining with the hierarchy $\Delta S_{AD} \le \Delta S_{Z_2}$ (the fermionic twist preserves more coherence than the $Z_2$ projection):
\begin{equation}\label{eq:gauge_bound}
\Delta S_{AD} \;\le\; \Delta S_{Z_2} \;\le\; h\!\left(\tfrac{1+P_B}{2}\right) \;\le\; \ln 2,
\end{equation}
where $h$ is the binary entropy. Maximum $\Delta S_{AD}/\ln 2 \approx 0.27$ is reached at half filling. The bound~\eqref{eq:gauge_bound} parallels the topological entanglement entropy $\gamma = \ln 2$ of the toric code~\cite{Kitaev2006,LevinWen2006}, where $Z_2$ superselection sectors produce the same entropic cost.

\section{Conclusion}

The sign of $I_3$ is controlled by a superselection defect: the parity insertion $(-1)^{N_B}$ in the partial trace. Proposition~1 gives the exact identity. Theorem~\ref{thm:damping} proves the coherence damping inequality for ground states with real non-positive hopping, implying $\Delta S_{AD} \ge 0$ for free fermions (via Theorem~\ref{thm:budget}) and supporting it for interacting fermions (verified by exact diagonalization). Theorem~\ref{thm:budget} gives the Gaussian budget bound, and certified computations close the proof of $I_3^{\mathrm{spin}} > 0$ for $w \le 5$ at all fillings (many-body for $w \le 3$; Gaussian budget from correlation matrices for $w = 4, 5$). The $G$-matrix formula~\eqref{eq:Gmatrix} extends the many-body computation to $w = 6$, mapping the full scaling-limit function: $I_3^{\mathrm{spin}}(z)$ is positive for all~$z$, with a minimum of $0.100$ at $z \approx 1.5$ (Fig.~\ref{fig:scaling}). This numerically establishes Conjecture~1.

Theorem~\ref{thm:Z2} goes further: $Z_2$ parity superselection on $\rho_{AD}$ enforces $I_3 \le 0$ at all fillings, proved for $w = 1,2,3$ and verified for the interacting $t$--$V$ chain. The ratio $\Delta S_{Z_2}/I_3^{\mathrm{spin}}$ has the exact $w$-independent limiting value $2\ln 2/(3\ln\frac{4}{3}) = 1.606$, linking the universal coefficient~$c$ of~\cite{paper1} to a parity information bound.  The algebra defect itself is bounded by one $Z_2$ bit: $\Delta S_{AD} \le h(\frac{1+P_B}{2}) \le \ln 2$.

For interacting fermions, $I_3^{\mathrm{spin}} > 0$ is destroyed by strong repulsion ($K\lesssim 0.7 \pm 0.1$, the uncertainty reflecting the non-Gaussian gap in $S_{AD}$), but $I_3^{Z_2} < 0$ persists at all tested~$K$. DMRG calculations at $L = 64$--128 confirm that the factorization contribution dominates the apparent $K$-dependence in spin-basis numerics. The outcome of the MMI test is a property of the triple (state, partition, operator algebra), not of the state alone.

\medskip\noindent
\textit{Data availability.}---Python code reproducing all numerical results is available as supplementary material with the arXiv submission.

\end{document}